\documentclass[journal]{IEEEtran} 
\ifCLASSINFOpdf

\else
\fi
\usepackage{subfig}
\usepackage{array}
\usepackage{amsmath} 
\usepackage{graphicx,wrapfig,lipsum} 
\usepackage{rotating}
\usepackage{amsmath, amssymb, amsthm}
\usepackage{stmaryrd}
\usepackage{tikz}
\usepackage{verbatim}
\usepackage{url}
\usepackage[utf8]{inputenc}
\usepackage[english]{babel}
\newtheorem{theorem}{Theorem}[]

\newtheorem{lemma}[]{Lemma}
\usepackage{multicol}
\usepackage{xr-hyper}

\usepackage{algorithm}
\usepackage{algorithmicx}
\usepackage{algpseudocode}
\usepackage{comment}

\usepackage{algpseudocode}
\usepackage{algcompatible}
\usepackage{scalerel}
\usepackage{multicol}
\usepackage{setspace}
\newtheorem{definition}{Definition}
\usepackage[noadjust]{cite}
\usepackage{bbm}
\usepackage[normalem]{ulem}
\usepackage{color}
\usepackage{dsfont}
\usepackage{bm}
\usepackage[nolist,nohyperlinks]{acronym}
\usepackage{graphicx}
\graphicspath{ {./images/} }
\usetikzlibrary{automata}
\usetikzlibrary{arrows}
\allowdisplaybreaks
\usepackage{tikz}
\usetikzlibrary{automata,positioning}

\usepackage{verbatim}
\usepackage[T1]{fontenc}
\usepackage[utf8]{inputenc}
\usepackage{authblk}

\DeclareMathOperator*{\argmax}{arg\,max}

\title{Joint Power and Subcarrier Allocation in Multi-Cell Multi-Carrier NOMA}
\author{%
{Subhankar Banerjee, Chung Shue Chen,~\IEEEmembership{Senior Member,~IEEE},  Marceau Coupechoux, and Abhishek Sinha} 
\thanks{S. Banerjee is with University of Maryland, College Park, MD 20742, USA.
}
\thanks{C. S. Chen is with Bell Labs, Nokia Paris-Saclay, 91620 Nozay, France.
}%
\thanks{M. Coupechoux is with LTCI, Telecom Paris, Institut Polytechnique de Paris, France.
}%
\thanks{A. Sinha is with the Tata Institute of Fundamental Research, Mumbai, India.
}%
\vspace{-0.9cm} 
\thanks{Email address: 
sbanerje@umd.edu, 
chung\_shue.chen@nokia-bell-labs.com, marceau.coupechoux@telecom-paris.fr, abhishek.sinha.tifr@gmail.com}
}

\begin{document}

\maketitle
\begin{abstract}
Non-orthogonal multiple access (NOMA) is a technology proposed for next generation cellular networks because of its high spectral efficiency and enhanced user connectivity. However, 
in the literature the optimal joint power and sub-carrier allocation for NOMA has been proposed for single cell only. Consequently, a global optimal algorithm for the joint power and sub-carrier allocation  for NOMA system in multi-cell scenario is still an open problem. In this work, we propose a polyblock optimization based algorithm for obtaining a global optimal solution. It has reduced complexity due to a necessary and sufficient condition for feasible successive interference cancellation (SIC). Besides, we can adjust its optimization approximation parameter to serve as benchmark solution or to offer suitable practical solution for multi-cell
multi-carrier NOMA systems. Numerical studies have shown its effectiveness.  
\end{abstract}
\vspace{-0.3cm}  

\section{Introduction}
Orthogonal frequency division multiple access (OFDMA) is the downlink multiplexing scheme adopted by the 5G New Radio. 
In each cell, every sub-carrier is allocated to at most one user such that 
intra-cell interference is almost suppressed.  
However, OFDMA is known to be sub-optimal in terms of spectral efficiency \cite{tse2005fundamentals}. Power domain 
non-orthogonal multiple access (NOMA) is a capacity-achieving multiple access scheme based on successive interference cancellation (SIC) that has been proposed for 
future mobile 
networks with the goal of improving the spectral efficiency.  
In NOMA, unlike orthogonal multiple access (OMA) such as OFDMA, each sub-carrier can be allocated to possibly more than one user; 
multiple users with diverse power levels can be accommodated within the same resource block simultaneously with the aid of advanced transmission scheme (superposition coding) and progressive reception technique (SIC) \cite{ding2017survey}.  
  
The intra-cell interference occurs due to the fact that multiple users in the same cell can be allocated to the same sub-carrier. Note that the resource allocation optimization problem for single cell NOMA has been addressed extensively in the literature, see for example 
\cite{parida2014power, sun2016optimal, salaun2020joint} 
and the references therein.  
In this paper, we focus on 
the resource allocation optimization problem for multi-cell NOMA. 
We consider a frequency reuse~1 system, i.e., every cell shares the same sub-carriers when performing NOMA, resulting in inter-cell interference within the system. The presence of inter-cell interference with intra-cell interference in such multi-cell NOMA system makes the resource allocation problem more difficult. 

In 
\cite{nguyen2017precoder}, the authors addressed  the uplink precoder design optimization problem for multi-cell MIMO-NOMA system and performed sum rate maximization using an
approximate algorithm. 
In 
\cite{qian2009mapel, kim2015sum, sun2016optimal}, the authors 
employed monotonic optimization based 
methods which 
aim to reach a 
global optimal solution. 
The paper \cite{qian2009mapel} addressed  
the global optimal power control problem in wireless networks over 
multiple interfering 
links. 
In 
the paper \cite{kim2015sum}, the power and sub-carrier allocation problem for sum-rate maximization in multi-cell OFDMA systems has been investigated. 
Note that in \cite{sun2016optimal, salaun2020joint},  
the optimal power
and sub-carrier allocation in NOMA has been formulated, however for a single cell. 

In \cite{8352643}, the authors consider a multi-cell single carrier CoMP-NOMA system and provide a 
power allocation solution 
among all the cells. 
In \cite{fu2017distributed}, 
the power minimization problem for a downlink NOMA multi-cell system subject to 
user minimum data rate 
requirement is addressed. 
However, 
the sum rate 
maximization problem is not investigated. 
In \cite{yang2018power}, 
heuristic algorithms for the power allocation in multi-cell multi-carrier NOMA systems for 
sum power minimization as well as sum rate maximization have been devised. 

To the best of our knowledge, so far no one has proposed a global optimal solution for the multi-cell multi-carrier NOMA join power and sub-carrier allocation problem for the sum rate maximization. Here, we aim to provide an optimal algorithm to solve the above problem under multi-cell NOMA. We propose a polyblock optimization based algorithm, which has been used extensively in many resource allocation problems, see for example \cite{qian2009mapel, kim2015sum}. We extend this method to solve NOMA resource allocation optimization problem in multi-cell setup. 
{In the meantime, we derive the necessary and sufficient condition for feasible SIC in multi-cell NOMA, which can be used for the class of problems to reduce the optimization complexity.} 


\section{System Model}
\label{s_model}

In this work, we 
consider a 
multi-cell multi-carrier NOMA 
downlink 
system, which consists of $K$ base stations (BS) denoted by the set $\mathcal{K}$. The set of users served by a BS $k\in\mathcal{K}$ is denoted as $\mathcal{M}_k$ (with cardinality $M_k = |\mathcal{M}_k|$). There is a total number of $L$ sub-carriers in the network denoted by the set $\mathcal{L}$. Note that we consider a frequency reuse~1 system such that 
all the sub-carriers are available to all the BS. 
If a sub-carrier $l$ is allocated 
to user $u$ by BS $k$, where $u\in{\mathcal{M}_{k}}$, we set $a_{k,u}^l=1$, and $a_{k,u}^l=0$ otherwise. Let $p_{k,u}^l\geq 0$ be the power allocated by BS $k$ to user $u\in\mathcal{M}_k$ on sub-carrier $l$ and $p_k^l$ be the total power transmitted by BS $k$ on sub-carrier $l$, i.e., $p_k^{l} = \sum_{u\in{\mathcal{M}_k}} p_{k,u}^{l}$. 
We 
define 
a vector $\mathbf{a}$ 
as follows\footnote{First, we fill all the entries of the first BS, followed by the entries of the second BS and 
continue until the last BS. 
Among the entries of 
each BS, first 
we fill the entries corresponding to the first sub-carrier, followed by the entries of the second sub-carrier and continue until the $L$-th sub-carrier.
Among the entries of 
each sub-carrier of a BS, we fill the entries corresponding to all the users in that BS.}:\\ 
\vspace{-0.5cm}
\begin{eqnarray}
\mathbf{a} &=& (
a_{1,1}^1,\cdots,a_{1,M_1}^1,
\cdots,a_{1,M_1}^2, \cdots,
\cdots,a_{1,M_1}^L, 
\cdots,\notag \\
&& a_{K,1}^1,\cdots,a_{K,M_K}^1,
\cdots,a_{K,M_K}^2, \cdots,
\cdots,a_{K,M_K}^L
). \notag
\end{eqnarray}
We have
$a_{k,u}^{l} \in \{0,1\}$, for all ${k\in{\mathcal{K}}},u\in{\mathcal{M}_{k}},l\in{\mathcal{L}}$.
We 
define a vector $\mathbf{p}$ and order the $p_{k,u}^{l}$, for all ${k\in{\mathcal{K}}},u\in{\mathcal{M}_{k}},l\in{\mathcal{L}}$, in the same manner. 
Both 
$\mathbf{a}$ and $\mathbf{p}$ 
have length equal to 
$\sum_{k=1}^{K} M_k L$.  


We 
consider 
the use of a SIC based receiver \cite{ding2017survey}. Such a receiver is characterized by 
its decoding order of the received signals. 
The decoding order among the users of BS $k$ on sub-carrier $l$ is defined by the vector $\pi_{k}^{l} = (\pi_{k}^{l}(1), \pi_{k}^{l}(2), \cdots, \pi_{k}^{l}(M_k))$, where $\pi_k^l(i)$ is the $i$-th user to be decoded. In particular, user $\pi_k^l(i)$ is able to decode and subtract the signal of users $\pi_k^l(1)$ to $\pi_k^l(i-1)$ and treats the signal from users $\pi_k^l(i+1)$ to $\pi_k^l(M_k)$ in cell $k$ on sub-carrier $l$ as interference. 
Notice that 
there is also inter-cell interference, which should be taken into account when other cells are using the same sub-carrier.
The vector $\pi_{k}^{l}$ can thus be seen as a function which maps the decoding order to the user index. Using this formalism, $(\pi_{k}^{l})^{-1}$ is the inverse function, which maps the user index to the corresponding decoding order.

With 
the above notations, the signal-to-interference-plus-noise ratio (SINR) for a user $u$ served by a BS $k$ on sub-carrier $l$ can be written as:
\begin{equation}\label{eq:gamma}
    \gamma_{k,u}^l=\frac{g_{k,u}^l
    p_{k,u}^l}{\sum_{i=(\pi_{k}^{l})^{-1}(u) + 1}^{M_{k}} g_{k,u}^{l} p_{k,\pi_{k}^{l}(i)}^{l} \hspace{-0.1cm} +  \sum_{j\in{\mathcal{K}}\backslash{\{k}\}} g_{j,u}^{l} p_{j}^{l} \hspace{-0.1cm} + N_{k,u}^{l}}
\end{equation}
where $g_{k,u}^l$ is the
link gain between BS $k$ and user $u$ on sub-carrier $l$, and $N_{k,u}^{l}$ is the power of the noise for user $u$ in cell $k$ and sub-carrier $l$.  
For the simplicity of discussion, in this paper, we would 
assume that the noise power is constant across users, sub-carriers and 
BS, i.e., $N_{k,u}^l= N$. 


We construct 
a vector $\boldsymbol{\gamma}$ from the $\gamma_{k,u}^l$
in the same manner as we did for 
$\mathbf{a}$ and $\mathbf{p}$. 
We thus have the three following vectors: 
\begin{eqnarray}\label{power:vec}
\mathbf{p}&=&[p_{k,u}^l]_{k\in\mathcal{K},l\in\mathcal{L},u\in\mathcal{M}_k}, \\
\mathbf{a}&=&[a_{k,u}^l]_{k\in\mathcal{K},l\in\mathcal{L},u\in\mathcal{M}_k}, \\
\boldsymbol{\gamma}&=&[\gamma_{k,u}^l]_{k\in\mathcal{K},l\in\mathcal{L},u\in\mathcal{M}_k}. 
\end{eqnarray} 

The maximum 
transmit power of a BS $k$ on sub-carrier 
$l$ is denoted by $\bar{p}_k^l$ 
such that we have the following constraint:
\begin{equation} \label{eq:scpowerconst}
    0\leq \sum_{u\in\mathcal{M}_k}p_{k,u}^l\leq \bar{p}_k^l, \forall k\in\mathcal{K}, l\in\mathcal{L}.
\end{equation}
We denote the maximum total transmit power of   
a BS $k$ 
for all the sub-carriers by $\bar{p}_k$, so that we have the following cellular power constraint \cite{salaun2020joint}: 
\begin{equation} \label{eq:bspowerconst}
    0\leq \sum_{l\in\mathcal{L}} \bar{p}_k^l \leq \bar{p}_k, \forall k\in\mathcal{K}.
\end{equation}

Because of SIC practical constraints due to decoding complexity and potential error propagation, 
we assume that there is a limitation on the maximum number of users that we can multiplex  
in 
each sub-carrier,  
denoted by $M$, i.e., 
\begin{equation} \label{eq:sicconst}
    \sum_{u\in\mathcal{M}_k}a_{k,u}^l \leq M, \forall k\in\mathcal{K}, l\in\mathcal{L}.
\end{equation}


 
\subsection{Conditions for Feasible SIC in Multi-cell NOMA}
\label{sub_SIC_Constraints}

We consider a fixed rule for the SIC decoding order 
and follow the same SIC ordering as it is done for single cell NOMA\footnote{This means that we do not optimize the SIC ordering for multi-cell. This aspect is known to be an open problem and left for future work.} 
in the literature  \cite{tse2005fundamentals}: users are sorted in the increasing order of their 
link gains;
a user with the 
smallest 
link gain is decoded first, 
whereas 
a user with the 
largest link gain
is decoded at the end of the decoding process. That is, 
a weak user (a user with lower link gain) decodes its signal 
and treats all other signals as interference, while a strong user 
(with higher link gain) can first 
decode a weak user's signal to remove it and then decodes its own data. 

Nevertheless, a constraint arises on the SIC ordering, which is specific to 
multi-cell scenario and 
a condition for SIC to be feasible. 
Consider the following example, a system with two cells and each cell with two users. Let's call the two users in cell $1$ as user $1$ and user $2$, while the two users in cell $2$ as user $3$ and user $4$. 
We do our analysis 
by considering that they are co-channel interferers (says $l = 1$).  
The power allocated for these users are $p_{1,1}^{1}$, $p_{1,2}^{1}$, $p_{2,3}^{1}$ and $p_{2,4}^{1}$, respectively. 
Consider that $g_{1,2}^{1} > g_{1,1}^{1}$ and $g_{2,4}^{1}>g_{2,3}^{1}$. 
Before SIC, the SINR for user $1$ is $\frac{g_{1,1}^{1} p_{1,1}^{1}}{g_{1,1}^{1} p_{1,2}^{1} + g_{2,1}^{1} p_{2}^{1}+N}$ and the 
SINR for 
user $2$ is 
$\frac{g_{1,2}^{1} p_{1,2}^{1}}{g_{1,2}^{1} p_{1,1}^{1} + g_{2,2}^{1} p_{2}^{1}+N}$, according to  \eqref{eq:gamma}. 
After SIC, 
the 
SINR for 
user $1$ (i.e., weak user) 
would be the same, 
whereas the SINR for 
user $2$ (i.e., strong user) 
would become $\frac{g_{1,2}^{1} p_{1,2}^{1}}{g_{2,2}^{1} p_{2}^{1} + N}$ since 
user $2$ can first decode 
user $1$'s signal and then remove it. 
However, it should be noted that user $2$ can only decode 
user $1$'s signal from the received 
signal if and only if the following condition holds: \vspace{-0.2cm}
\begin{equation} 
    \frac{g_{1,2}^{1} p_{1,1}^{1}}{g_{1,2}^{1} p_{1,2}^{1} + g_{2,2}^{1} p_{2}^{1}+N} > \frac{g_{1,1}^{1} p_{1,1}^{1}}{g_{1,1}^{1} p_{1,2}^{1} + g_{2,1}^{1} p_{2}^{1}+N}. 
\end{equation}
By simplifying the above expression, we can obtain  $(g_{1,2}^{1} g_{2,1}^{1} - g_{1,1}^{1} g_{2,2}^{1}) p_{2}^{1} + (g_{1,2}^{1} - g_{1,1}^{1}) N \geq 0$. 
We state the above 
observation for general setting as the following theorem.

 \begin{theorem} \label{th:sicconstraint}
 For any two users in $\mathcal{M}_k$ served by BS $k$ on a sub-carrier $l$ such that $g_{k,1}^l<g_{k,2}^l$, a necessary and sufficient condition for user 2 (a strong user) to be able to remove user~1's signal (a weak user) is given by:  
 \begin{equation}\label{eq:sicconst1}
  \left( \sum_{i\in\mathcal{K}\backslash\{k\}}\hspace{-0.1cm} \left(g_{k,2}^l g_{i,1}^l - g_{k,1}^l g_{i,2}^l\right)p_i^l \right) + \left(g_{k,2}^{l} - g_{k,1}^{l}\right)N  \geq 0. 
 \end{equation} 
 \end{theorem}
 \begin{proof}
\vspace{-0.2cm} 
See Appendix~\ref{app:sicconstraint}. 
 \end{proof}

Following the result of Theorem~\ref{th:sicconstraint}, we can see that 
for sub-carrier $l$ and BS $k$, there are $M_k-1$ corresponding constraints (necessary and sufficient condition) for SIC to be feasible and 
hence a total of $\sum_{k=1}^{K} L(M_k -1)$ constraints for a multi-cell multi-carrier NOMA system. 
Consider that $\mathbf{p}_{1}$ and $\mathbf{p}_{2}$ are two power vectors as in (\ref{power:vec}). 
If $\mathbf{p}_{1}$ satisfies (\ref{eq:sicconst1}), then any vector $\mathbf{p}_{2}$, which is coordinate-wise lesser than $\mathbf{p}_{1}$, may not satisfy (\ref{eq:sicconst1}). The second term in (\ref{eq:sicconst1}) is always 
non-negative since $N$ is positive and $g_{k,2}^{l} \geq g_{k,1}^{l}$, with the chosen decoding order. If $(g_{k,2}^{l} g_{i,1}^{l} - g_{k,1}^{l} g_{i,2}^{l}) \geq 0,  \forall{i\in{\mathcal{K}\backslash\{k\}}}$, then every $\mathbf{p}\geq \mathbf{0}$ satisfies (\ref{eq:sicconst1}). Note that $\mathbf{0}$ is a zero vector of the same size of $\mathbf{p}$  
and the inequality is coordinate-wise. 
In Section~\ref{Simulation}, 
{we conduct a simulation to show that the assumption 
$(g_{k,2}^{l} g_{i,1}^{l} - g_{k,1}^{l} g_{i,2}^{l}) \geq 0$ is 
generally true with high probability in practical reference scenarios.} 
\vspace{-0.2cm}

\subsection{Problem Formulation}
\label{sub_Problem_Formulation}

The optimization problem considered in this work is formulated as follows:
\begin{eqnarray}\label{eq:opt_sol}
\max_{\mathbf{a},\mathbf{p}} & \displaystyle\sum_{k\in{\mathcal{K}}} \sum_{u\in{\mathcal{M}_k}}\sum_{l\in\mathcal{L}} a_{k,u}^l \log(1+\gamma_{k,u}^l) 
\vspace{-0.3cm}
\end{eqnarray}
\begin{eqnarray}
\mbox{subject to} & (\ref{eq:scpowerconst}), (\ref{eq:bspowerconst}), (\ref{eq:sicconst}), (\ref{eq:sicconst1})  \notag
\end{eqnarray}
where 
the constraint (\ref{eq:scpowerconst}) ensures that powers are non-negative and the per sub-carrier power constraint is met, 
constraint (\ref{eq:bspowerconst}) is the per BS power constraint, and 
constraints (\ref{eq:sicconst}) and (\ref{eq:sicconst1}) are due to SIC. 
The objective function (\ref{eq:opt_sol}) is the sum rate for multi-cell NOMA systems. 
Since $a_{k,u}^l\in\{0,1\}$, we can also 
write 
(\ref{eq:opt_sol}) as:
\begin{eqnarray} \label{eq:opt_sol1}
\max_{\mathbf{a},\mathbf{p}} & \log & \prod_{k\in{\mathcal{K}}} \prod_{u\in{\mathcal{M}_k}}\prod_{l\in\mathcal{L}}  \left(1+ a_{k,u}^l \gamma_{k,u}^l\right)
\vspace{-0.3cm}
\end{eqnarray}
\begin{eqnarray}
\mbox{subject to} & (\ref{eq:scpowerconst}), (\ref{eq:bspowerconst}), (\ref{eq:sicconst}),  (\ref{eq:sicconst1}). \notag
\end{eqnarray}
\vspace{-0.6cm}

\section{Optimal Power and Sub-carrier Allocation (OPSA)}
\label{p_Algorithm}

In this section, we propose an algorithm to solve the formulated problem. We start with the preliminary definitions and results.
\begin{definition}[Normal Set]
A set ${\mathcal{A}} \subset \mathbb{R}_{+}^{N}$ is a normal set if for any ${\mathbf{x}\in{\mathcal{A}}}$, $\{{\mathbf{x}'|\mathbf{x}'\leq{\mathbf{x}}}\}\subset\mathcal{A}$, where inequality is component-wise. 
\vspace{-0.1cm}
\end{definition}
Note that the intersection and the union of normal sets are normal. 
\begin{definition}[Box]
Given any vector $\mathbf{x} \in {\mathbb{R}_+^{n}}$, the hyper rectangle $[\mathbf{0},\mathbf{x}]=\{\mathbf{v}|\mathbf{0}\leq \mathbf{v} \leq \mathbf{x}\}$ is called a box.
\end{definition}
Note that a box is a normal set. 

\begin{definition}[Monotonic Optimization]
A monotonic optimization problem is a class of optimization problems that have the following formulation: \vspace{-0.2cm}
\begin{equation} \label{eq:monotonic}
    \max_{\mathbf{x}} \hspace{0.1 cm} f(\mathbf{x})
\end{equation}
\vspace{-0.2cm}
\begin{equation*}
    \mbox{subject to} \hspace{0.1 cm} \mathbf{x}\in{\mathcal{A}}
\end{equation*}
where $f$ is an increasing function, $\mathcal{A}$ is a normal set.
\end{definition}

\vspace{-0.3cm}
\subsection{Monotonic Optimization}
\label{sub_Monotonic_Optimization}

Following the approach proposed in \cite{qian2009mapel}, we replace the expression $(1+ a_{k,u}^l \gamma_{k,u}^l)$ in (\ref{eq:opt_sol1}) by 
a new variable $z_{k,u}^l$ and  re-write 
(\ref{eq:opt_sol1}) as: \vspace{-0.1cm}
 \begin{eqnarray} \label{eq:oth_opt}
\max_{\mathbf{z}} & \log & \prod_{k\in{\mathcal{K}}} \prod_{u\in{\mathcal{M}_k}}\prod_{l\in\mathcal{L}}  z_{k,u}^l \label{eq:zlogfunction}
\end{eqnarray}
\begin{eqnarray}
\vspace{-0.3cm}
\mbox{subject to} & \mathbf{z} \in \mathcal{Z} \notag
\end{eqnarray}
where $\mathbf{z}$ is the vector that comprises all the $z_{k,u}^{l}$ and 
sometimes we may simply write $z_i$ for its $i$-th component, 
where $i=1, 2, \ldots, MKL$. We call this vector as the SINR vector. Let $\mathcal{Z}$ be the set of all possible $\mathbf{z}$. Let us assume that $\mathbf{z}^{*}$ is an SINR vector solution of (\ref{eq:oth_opt}), and denote $\mathbf{a}^{*}$ and $\mathbf{p}^{*}$ are the corresponding power and SINR vectors  respectively. Let us denote $N_{1}$ be the number of components of $\mathbf{z}^{*}$ which take the value greater than $1$, 
{where $0 \leq N_{1} \leq MKL$.} When a component $\mathbf{z}^{*}$ is $1$, the corresponding components of $\mathbf{a}^{*}$ and $\mathbf{p}^{*}$ take the value $0$. When a component of $\mathbf{z}^{*}$ is greater than $1$, the corresponding component of $\mathbf{a}^{*}$ would take the value $1$. We rewrite $1+\gamma_{k,u}^{l}$ as $\frac{f_{i}(\mathbf{p})}{g_{i}(\mathbf{p})}$, where $f_{i}$ and $g_{i}$ represent the linear functions of $\mathbf{p}$ corresponding to the $i$-th component of $\mathbf{z}$. Thus, we get $N_{1}$ linear equations of the form $z_{i} g_{i} - f_{i} = 0$, $\forall i$, such that $z_{i} > 1$. The coefficients of $g_{i}$ and $f_{i}$ are all independent random channel gains and we can show that with probability $1$, all $N_{1}$ linear equations are linearly independent, implying that there is a unique $\mathbf{p}$ for every $\mathbf{z}$. The corresponding $\mathbf{p}$ for every $\mathbf{z} \in \mathcal{Z}$ must satisfy (\ref{eq:scpowerconst}), (\ref{eq:bspowerconst}), (\ref{eq:sicconst}),  (\ref{eq:sicconst1}).
Thus, from the optimal solution $\mathbf{z}^{*}$ of (\ref{eq:oth_opt}), we get the optimal user allocation $\mathbf{a}^{*}$ and optimal power allocation $\mathbf{p}^{*}$. We state and prove the following lemma.
\begin{lemma}\label{lemma:1}
If there exists two SINR vectors $\mathbf{z}_{1}$ and $\mathbf{z}_{2}$ such that $\mathbf{z}_{1} \leq \mathbf{z}_{2}$, then the corresponding power vectors $\mathbf{p}_{1}$ and $\mathbf{p}_{2}$ must satisfy the relation $\mathbf{p}_{1}\leq \mathbf{p}_{2}$.
\end{lemma}
\begin{proof}
See Appendix~\ref{app:e}. 
\end{proof}

\begin{lemma}\label{lemma:2}
The set of SINR vectors, $\mathcal{Z}$ corresponding to the power vectors which satisfies, (\ref{eq:scpowerconst}), (\ref{eq:bspowerconst}), (\ref{eq:sicconst}), is a normal set.
\end{lemma}
\begin{proof}
See Appendix~\ref{app:f}. 
\end{proof}


We propose 
an optimal algorithm for the problem, when (\ref{eq:sicconst1}) holds true for any $\mathbf{p}\geq\mathbf{0}$. As we have seen in the discussion after Theorem~\ref{th:sicconstraint} and also from simulation result Fig.~\ref{cell_radi}, this is generally true in a practical setting with high probability. Thus, from Lemma~\ref{lemma:2}, we say that $\mathcal{Z}$ is a normal set. 

Let us 
define $f(\mathbf{z})\triangleq  \log \Pi_{k\in{\mathcal{K}}} \Pi_{u\in{\mathcal{M}_k}}\Pi_{l\in\mathcal{L}}  (z_{k,u}^l)^{w_{u}^k}$. 

\begin{lemma} \label{lem:fLC}
The function $f(\mathbf{z})$
is Lipschitz continuous.
\end{lemma}
\begin{proof}
See Appendix~\ref{app:lip}. 
\end{proof}

The problem \eqref{eq:oth_opt} is a monotonic optimization problem as $\mathcal{Z}$ is a normal set and $f(\mathbf{z})$ is an increasing function of $\mathbf{z}$. To solve 
\eqref{eq:oth_opt},  
we employ the 
outer polyblock approximation algorithm  \cite{hadjisavvas2006handbook}. 
We solve (\ref{eq:oth_opt}) in two steps. First, we find the optimal sub-carrier allocation then the optimal power allocation by employing the polyblock approximation algorithm. Next, we propose a theorem which is essential to find the optimal sub-carrier allocation.
\begin{theorem}\label{th:2}
If the power distribution over all the $K$ BS for the $l$-th sub-carrier is $(p_{1}^{l}, \cdots, p_{K}^{l})$, then the optimal choice is to assign all the power given to 
each BS to a user with the highest link gain   
in that sub-carrier, i.e., let $u^\star=\arg\max_u g_{k,u}^l$ (with any tie-breaking rule), then
\begin{eqnarray} \label{eq:red_poly_a}
a_{k,u}^l &=& \begin{cases}
1 & \text{if } u=u^\star \\
0 & \text{otherwise}
\end{cases}
\end{eqnarray}
and 
\begin{eqnarray} \label{eq:red_poly_p}
p_{k,u}^l &=& \begin{cases}
p_{k}^l & \text{if } u=u^\star \\
0 & \text{otherwise.}
\end{cases}
\end{eqnarray}
\end{theorem}
\begin{proof} 
See Appendix~\ref{app:poly_reduced}. 
\end{proof}

Using Theorem \ref{th:2}, we can get the optimal sub-carrier allocation. In each sub-carrier and in each BS, we choose the user with the best channel gain to get the optimal sub-carrier allocation $\mathbf{a}^{*}$. To obtain the coordinates of $\mathbf{a^{*}}$, which correspond to the best channel gain user in each sub-carrier and in each BS, we put $1$, while to the other coordinates, we put $0$.

Now, we apply the polyblock approximation algorithm to get an optimal power allocation $\mathbf{p}^{*}$ (see Algorithm~\ref{alg:IA}). 
We first create an initial SINR vector (see line 3), where the maximum sub-carrier power is allocated to the active user 
in the 
sub-carrier allocation 
while interference is ignored. 
This vector is clearly infeasible and a box defined by this vertex covers the feasible set of SINR vectors defined by power constraints  \eqref{eq:scpowerconst} and \eqref{eq:bspowerconst}. We project (see line 8) this vector on the feasible set using Dinkelbach algorithm \cite{zappone2015energy} (see Algorithm~\ref{alg:Dinkelbach}). From the projected vector, we construct new 
SINR vectors (see line 9) and then remove the parent vector 
(see line 10).  
We repeat 
the above procedure until the objective value of the 
projected 
is $\epsilon$-close to the objective value of the projected vector (
see line 15), where $\epsilon$ is an input parameter of the algorithm.
Since the function $f(\mathbf{z})$ is a Lipschitz continuous function (see Lemma~\ref{lem:fLC}), the polyblock algorithm converges to an $\epsilon$-optimal solution in a finite number of iterations \cite{zhang2013monotonic}. 

\begin{algorithm}[t]
\begin{algorithmic}[1]
        \State \textbf{Inputs}: $\epsilon$, $g_{k,u}^l$, $\bar{p}_{k}$, $\bar{p}_{k}^{l}$, $\forall l\in\mathcal{L}$, $\forall u \in \mathcal{M}_k$, $\forall k \in{\mathcal{K}}$, $N$, $M$
        \State \textbf{Define}:  $f^{\star}\gets -\infty$ 
         \State \textbf{Initialization} : For the optimal sub-carrier allocation $\mathbf{a}^{*}$, define an SINR vector $\mathbf{z}_{\mathbf{a}^{*}}$ such that $z_{k,u}^l\gets 1+g_{k,u}^l \bar{p}_k/N$ if $a_{k,u}^l=1$ and $z_{k,u}^l\gets 0$ otherwise.
         \State $\mathcal{V}\gets \{\mathbf{z}_{\mathbf{a}^{*} }\} $
         \State $f^{\star}\gets -\infty$ 
        \REPEAT 
        \State  
         $\mathbf{z} \gets 
         \arg\max_{\mathbf{\mathbf{z}_{\mathbf{a}}}\in{\mathcal{V}}}  f(\mathbf{z}_{\mathbf{a}})$
         \State [$\pi_z({\mathbf{z}}), \lambda$]$\gets \text{Algorithm}~\ref{alg:Dinkelbach}$
         \State $\mathbf{z}_i \gets \mathbf{z} - (\mathbf{z}-\pi_{z}(\mathbf{z}))\circ  \mathbf{e}_i$, $\forall i$, where $\mathbf{e}_i$ is the $i$-th basis vector, $\circ$ is the point-wise multiplication.
         \State $\mathcal{V}\gets \mathcal{V}-\{\mathbf{z}\}\cup_i \mathbf{z}_i$, $\forall {i}$, $ \mbox{subject to~}$ $a_i=1$, $p_i > 0$   \IF{$f(\pi_{z}({\mathbf{z}})) \geq f^{\star}$}
         \State $f^{\star}\gets f(\pi_{z}({\mathbf{z}}))$ 
         \State $\mathbf{z}^{\star}\gets \pi_{z}({\mathbf{z}})$
         \ENDIF
         \UNTIL$|f(\mathbf{z}) - f^{\star}| \leq \epsilon$
         \State \textbf{Output:} $\mathbf{z}^{\star}$
\end{algorithmic}
\caption{{Polyblock approximation algorithm}}
\label{alg:IA}
\end{algorithm}

\begin{algorithm}{}
\begin{algorithmic}[1]
        \State \textbf{Inputs}: $\mathbf{z}$, $\epsilon$, $g_{k,u}^l$, $\bar{p}_{k}$, $\bar{p}_{k}^{l}$, $\forall l\in\mathcal{L}$, $\forall u \in \mathcal{M}_k$, $\forall k \in{\mathcal{K}}$, $N$, $M$
        \State \textbf{Initialization}: $\mathbf{p} \gets \mathbf{0}$ 
        \State Compute $\mathbf{r}$ corresponding to $\mathbf{p}$, using \eqref{eq:ni} and \eqref{eq:di}
        \REPEAT
        \State $\lambda \gets \min_i{\left(\frac{n_i}{d_i z_i}\right)}$
        \State $\mathbf{p} \gets \argmax_{\mathbf{p}\in{\mathcal{P}}} \min_{i} n_{i} - \lambda d_i z_i$
        \State Compute $\mathbf{r}$ and $\mathbf{z}$ using $\mathbf{p}$ 
         \UNTIL{$\min_i {n_i - \lambda {d_i} z_i}\geq{0}$}
         \State \textbf{Output:}[$\lambda \mathbf{z} , \lambda]$ 
        
\end{algorithmic}
\caption{Dinkelbach algorithm}
\label{alg:Dinkelbach}
\end{algorithm}

Note that the projection of $\mathbf{z}$ in Algorithm~\ref{alg:IA} (see line 8) is performed by the Dinkelbach algorithm  
(see Algorithm~\ref{alg:Dinkelbach}). From a feasible power vector $\mathbf{p}$, we define a vector $\mathbf{r}$ of 
length $MKL$, whose components are fractions of the allocated powers, as follows: 
\begin{eqnarray}
\label{eq:ri}
r_i=n_i/d_i, \forall i,
\end{eqnarray}
where 
\begin{eqnarray}
& n_i = g_{k,u}^l p_{k,u}^l +   g_{k,u}^l \sum_{j\in{\mathcal{M}_{k}},j={\pi_{k}^{l}(u)+1}} p_{k,j}^l + 
~~~~~~ \nonumber \\ & ~~~~~~~~~~~~~~~~~~~~~~~~~~~~~~~
\sum_{i\in\mathcal{K}\backslash\{k\}} g_{i,u}^{l} p_{i}^{l}+N,  \label{eq:ni} \\   
& d_i =  g_{k,u}^l  \sum_{j\in{\mathcal{M}_{k}},j={\pi_{k}^{l}(u)+1}} p_{k,j}^l +
~~~~~~~~~~~~~~~~~~~ \nonumber \\ & ~~~~~~~~~~~~~~~~~~~~~~~~~~~~~~~
\sum_{i\in\mathcal{K}\backslash\{k\}} g_{i,u}^{l} p_{i}^{l}+N, \label{eq:di}
\end{eqnarray}
and $i$ is mapped to the triplet $(k,u,l)$ corresponding to 
BS $k$, user $u$ and sub-carrier $l$. 

\section{Simulation Result}
\label{Simulation}

In our numerical studies, we consider  two neighboring BS, with each BS having $2$ sub-carriers and $3$ users. The SIC system constraint parameter $M$ in (\ref{eq:sicconst}) is set to $2$. We consider hexagonal cell of radius equal to 100 meters and the users are dropped in each cell randomly following a uniform distribution. 
We follow the radio propagation model of \cite{Greentouch} with distance-dependent path loss $128.1 + 37.6 \log_{10} d$, where $d$ is the distance between a BS and the user. We consider that the sub-carrier bandwidth equals to 1 MHz while the noise spectral density is -174 dBm/Hz.   

To begin with, we show the empirical cumulative distribution function (CDF) of the quantity $(g_{k,2}^{l} g_{i,1}^{l} - g_{k,1}^{l}  g_{i,2}^{l})$ discussed in Theorem 1, for cell radii of 100, 200 and 500 meters. As shown in Fig.~\ref{cell_radi}, we can see that the probability that $(g_{k,2}^{l} g_{i,1}^{l} - g_{k,1}^{l} g_{i,2}^{l})\geq0$ is very high for these reference scenarios. 
Thus, we can 
suggest that with very high probability the constraint (\ref{eq:sicconst1}) 
can hold for every $\mathbf{p}\geq \mathbf{0}$.

In Fig.~\ref{fig:nomaproj}, we show the sum rate obtained by the proposed algorithm with $\epsilon$ set to $0.1$, $0.5$ and $1$, respectively. Meanwhile, we vary the maximum transmit  power per sub-carrier value  $\bar{p}_k^l$. Besides, we compare with the algorithm proposed in \cite{yang2018power}, which is a state-of-the-art heuristic optimization algorithm for multi-cell NOMA system. Result shows that the proposed algorithm can outperform the reference algorithm \cite{yang2018power} and provide more optimal solution. 

In Fig.~\ref{fig:nomaproj1}, we plot the average running time taken by the proposed algorithm in parallel to that taken by the algorithm of \cite{yang2018power} to compare their computational complexity.\footnote{{It is run on a common computer with Windows 10, 64-bit OS, Intel Core i5-4590 CPU at 3.30GHz with 16 GB of RAM installed.}} 
As expected, with a larger $\epsilon$ (which is the approximation parameter), the time taken by the proposed algorithm would be shorter and is getting closer to the time taken by the algorithm of \cite{yang2018power}. 
However, we can see that for 
$\epsilon = 0.1$ to $1$, overall the time complexity of the proposed algorithm is comparable to that of \cite{yang2018power}. 
Result shows that it would be suitable to set the proposed algorithm with a small $\epsilon$ such as 0.1 at the reasonable time cost for finding a more optimal solution to the sum rate maximization problem. 

\begin{figure}[t]
\centering
	\includegraphics[height = 6.8cm, width = 0.5\textwidth]{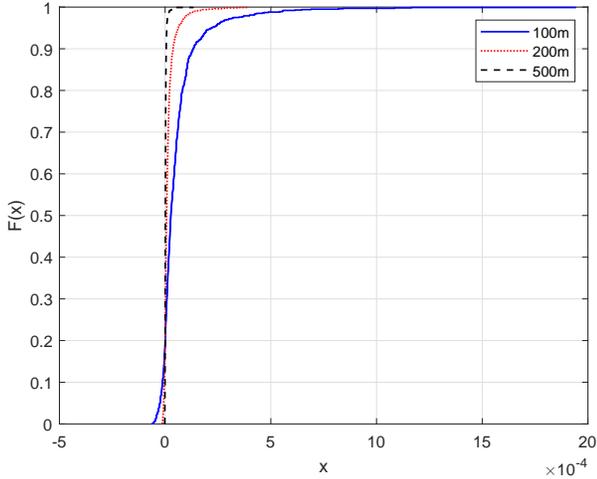}	
	\caption{CDF plot of $(g_{1,2}^l g_{2,1}^l - g_{1,1}^l g_{1,2}^l)$ with cell radius equal to 100, 200 and 500 meters, respectively.\vspace{-0.3cm}}
	\label{cell_radi}
\end{figure}

\begin{figure}[t]
	\centering
	\includegraphics[height = 6.8cm, width = 0.5\textwidth]{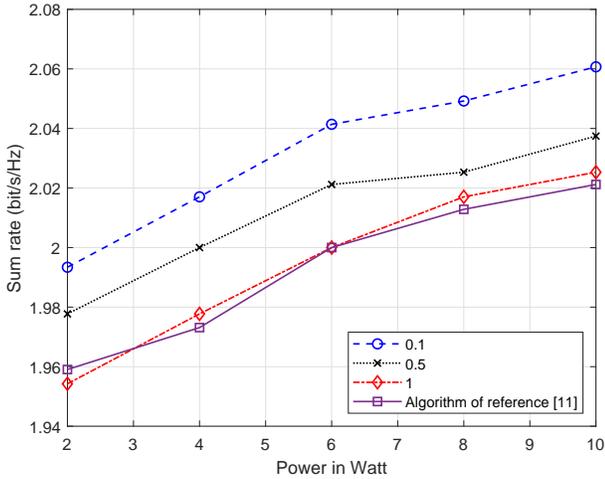}	
	\caption{Sum rate evaluation of NOMA system with varying maximum transmit power per sub-carrier.\vspace{-0.3cm} }
	\label{fig:nomaproj}
\end{figure}

\begin{figure}[t]
	\centering
	\includegraphics[height = 6.8cm, width = 0.5\textwidth]{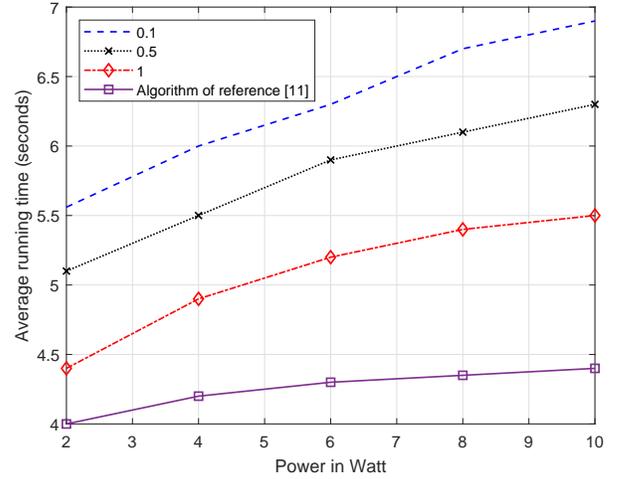}	
	\caption{Average running time of the proposed algorithm in comparison to that of reference [11].}  \vspace{-0.2cm} 
	\label{fig:nomaproj1}
\end{figure}

\section{Conclusion}
\label{s_Conclusion}

In this paper, we have proposed a global optimal algorithm for the joint power and sub-carrier allocation under multi-cell NOMA systems. Our scheme is based on the polyblock approximation algorithm. 
It has reduced complexity due to the obtained necessary and sufficient condition for feasible SIC.  
In practice, we can further adjust the approximation parameter to serve as benchmark solution or to provide suitable solution for solving the multi-cell multi-carrier NOMA resource allocation problem. 
Simulation result and the comparative study have also shown the effectiveness of the proposed scheme.

\bibliography{refer.bib}
\bibliographystyle{IEEEtran}


\newpage

\begin{appendices}
\label{Appendix}

\section{Proof of Theorem~\ref{th:sicconstraint}} \label{app:sicconstraint}

\begin{proof}
In the literature of power allocation for NOMA systems and also in this paper, if a BS assigns power $p$ to a user, then the achievable rate at the user is given by $\log\Big(1+\frac{g\cdot p} {N_{x} + I_{x}}\Big)$, where $I_{x}$ denotes all the interference seen by a user $x$ and $N_{x}$ is the power of an additive white Gaussian noise at user $x$. Let us consider two users indexed as $1$ and $2$, who are served by BS $k$ on sub-carrier $l$, and  $g_{k,1}^l< g_{k,2}^l$. The achievable rate at user 1 is given by $R_{1}=\log\Big(1+\frac{g_{k,1}^l p_{k,1}^{l}}{I_{1}+N_{1}}\Big)$. Similarly, the achievable rate at user 2 is given by $R_{2}=\log\Big(1+\frac{g_{k,2}^l p_{k,2}^{l}}{I_{2}+N_{2}}\Big)$. With SIC, denote the SINR of user 1 by $\gamma_1$ and the SINR at user 2 by  $\gamma_2$ such that 
\begin{equation*}
 \gamma_1=\frac{g_{k,1}^l p_{k,1}^l}{g_{k,1}^l \sum_{j\in{\mathcal{M}_{k}},j={\pi_{k}^{l}(1)+1}} p_{k,j}^l + \sum_{i\in\mathcal{K}\backslash\{k\}} g_{i,1}^{l} p_{i}^{l}+N}, 
\end{equation*}  
\begin{equation*}
 \gamma_2=\frac{g_{k,2}^l p_{k,1}^l}{g_{k,2}^l \sum_{j\in{\mathcal{M}_{k}},j={\pi_{k}^{l}(1)+1}} p_{k,j}^l + \sum_{i\in\mathcal{K}\backslash\{k\}} g_{i,2}^{l} p_{i}^{l}+N}. 
\end{equation*} 
As the SIC order is as such that user 2 should be able to decode the data of user 1, 
we know $\gamma_2$ has to be greater than or equal to $\gamma_1$, i.e., $\gamma_2-\gamma_1>0$, which leads to \eqref{eq:sicconst1} after simplification.
\end{proof}

\section{Proof of Lemma \ref{lemma:1}}\label{app:e}
\begin{proof}
 Let us assume that $\mathcal{I}$ be the set which stores those indices of $\mathbf{z}_{2}$ where it takes the value $1$, i.e., $i\in{\mathcal{I}}$ if $z_{2i}=1$. As,  $\mathbf{z}_{2}\geq \mathbf{z}_{1}$, the elements of $\mathbf{z}_{1}$ corresponding to the elements of $\mathcal{I}$ also take the value $1$, i.e., $z_{1i}=1,  \forall i\in\mathcal{I}$. Thus, the corresponding elements of both $\mathbf{p}_{1}$ and $\mathbf{p}_{2}$ take the value $0$ and for those elements the statement of lemma holds true.  We assume that $\mathcal{I}^{'}$ is the set which stores those indices of $\mathbf{z}_{2}$ where it take the value greater than $1$, i.e., $i\in{\mathcal{I}^{'}}$ if $z_{2i}>1$. In the following, we prove the statement of the lemma for these elements by contradiction. Let us assume that $\mathcal{I}^{''}$ be the set which stores the indices for which $p_{1i}\geq p_{2i}$ holds true. We define a real number $a$ and an integer $\alpha$, which are $\max_{i\in\mathcal{I}^{''}}\frac{p_{1i}}{p_{2i}}$ and $\argmax_{i\in\mathcal{I}^{''}}\frac{p_{1i}}{p_{2i}}$,  respectively. Now, let us assume that $\mathcal{I}^{'''}$ be the set which consists of all the indices of power vectors, for which the corresponding elements of the power vectors appear in the denominator of the $\alpha$-th element of SINR vector. We define another real number $b$ which is $\max_{i\in\mathcal{I}^{'''}} \frac{p_{1i}}{p_{2i}}$. Note that $a\geq b$. From the definition of $z_{i}$, we know $z_{i} = 1+ \gamma_{i}$. We define $int_{i}^{\mathbf{p}}$ as the summation of all interference signals term in $\gamma_{i}$ and $sig_{i}^{\mathbf{p}}$ as the desired signal term in $\gamma_{i}$ corresponding to power vector $\mathbf{p}$. We consider both the case of $b <1 $ and the case of $b\geq1$. When $b < 1$,  
 \begin{eqnarray}
 z_{1\alpha} &=& 1 + \frac{sig_{\alpha}^{\mathbf{p}_{1}}}{int_{\alpha}^{\mathbf{p}_{1}}+N} \nonumber \\
 &\overset{1} {\geq}& 1 + \frac{a \hspace{0.1 cm} sig_{\alpha}^{\mathbf{p}_{2}}}{b \hspace{0.1 cm} int_{\alpha}^{\mathbf{p}_{2}}+N} \nonumber \\
 &\overset{2} {\geq}& 1 + \frac{ sig_{\alpha}^{\mathbf{p}_{2}}}{ int_{\alpha}^{\mathbf{p}_{2}}+N} \nonumber \\
 &=& z_{2\alpha} 
 \end{eqnarray}
 which contradicts the assumption of the lemma that $\mathbf{z}_{2}\geq\mathbf{z}_{1}$. A  linear combination of all the elements of $\mathbf{p}_{2}$ corresponding to all the elements of $\mathcal{I}^{'''}$ is present at $int_{\alpha}^{\mathbf{p}_{2}}$ and from the definition of $b$ the first inequality holds true. The second inequality is true because of the fact that $\alpha\geq1$ and $b<1$. 
 Now we consider the case when $b\geq1$, 
 \begin{eqnarray}
 z_{1\alpha} &=& 1 + \frac{sig_{\alpha}^{\mathbf{p}_{1}}}{int_{\alpha}^{\mathbf{p}_{1}}+N} \nonumber \\
 &\overset{1} {\geq}& 1 + \frac{a \hspace{0.1 cm} sig_{\alpha}^{\mathbf{p}_{2}}}{b \hspace{0.1 cm} int_{\alpha}^{\mathbf{p}_{2}}+N} \nonumber \\
 &\overset{2} {\geq}& 1 + \frac{a \hspace{0.1 cm} sig_{\alpha}^{\mathbf{p}_{2}}}{b (\hspace{0.1 cm} int_{\alpha}^{\mathbf{p}_{2}}+N)} \nonumber \\
  &\overset{3} {\geq}& 1 + \frac{ \hspace{0.1 cm} sig_{\alpha}^{\mathbf{p}_{2}}}{\hspace{0.1 cm} int_{\alpha}^{\mathbf{p}_{2}}+N} \nonumber \\
  &=& z_{2\alpha}. 
 \end{eqnarray}
 Again it contradicts the assumption of the lemma. The second inequality and the third inequality hold true as $b\geq1$ and $\alpha\geq b$. This completes the proof. 
 \end{proof}

\section{Proof of Lemma \ref{lemma:2}}\label{app:f}
 \begin{proof}
It is to prove the lemma 
that if $\mathbf{z}_{1}\in{\mathcal{Z}}$, then every $\mathbf{z}_{2}$ which satisfies $\mathbf{z}_{2}\leq\mathbf{z}_{1}$ should also be in $\mathcal{Z}$. 
We have given an argument before that every power vector has an one-to-one correspondence with every SINR vector. Let us assume that the power vector corresponding to $\mathbf{z}_{1}$ is $\mathbf{p}_{1}$ and the power vector corresponding to $\mathbf{z}_{2}$ is $\mathbf{p}_{2}$. As  $\mathbf{z}_{1}\geq\mathbf{z}_{2}$, from Lemma~\ref{lemma:1} we can say that $\mathbf{p}_{1}\geq\mathbf{p}_{2}$. It is straightforward to see that if $\mathbf{p}_{1} \geq 0 $ and satisfies (\ref{eq:scpowerconst}), (\ref{eq:bspowerconst}) and (\ref{eq:sicconst}), then $\mathbf{p}_{2}$ also satisfies (\ref{eq:scpowerconst}), (\ref{eq:bspowerconst}) and (\ref{eq:sicconst}). Thus, $\mathbf{z}_{2}$ also is in $\mathcal{Z}$.
 \end{proof}
 
\section{{Proof of Lemma~\ref{lem:fLC}}}\label{app:lip}

\begin{proof}
To prove Lemma~\ref{lem:fLC}, it suffices to show that the function $f_n(\mathbf{x}) = \log (\Pi_{i=1}^{n} x_{i}^{w_{i}})$, where $\mathbf{x}=(x_i)_{i=1,\ldots,n}$ and $x_i\geq 1$, $\forall i$, is Lipschitz continuous. We prove this by mathematical induction. 
 
For $n=1$, let $x_1\geq 1$ and $x_{2}\geq 1$ be two scalars and a weight $w_1> 0$. Recall the  well-known logarithmic inequality $\log x \leq (x-1)$, for $x \geq 1$. Without loss of generality, assume that $x_{1}\geq  x_{2}$. Thus, we have:
 \begin{eqnarray}
 |f_1(x_1)-f_1(x_2)|& = & \left |\log \left (\frac{x_1}{x_2} \right) ^{w_1}\right | \notag \\
 & \leq & w_1 \left |\frac{x_1}{x_2}-1\right | \notag \\
 & = & \frac{w_1}{x_2}  |x_1-x_2| \notag \\
 & \leq & w_1 |x_2-x_1| \notag
 \end{eqnarray}

 We now assume that $f_{n-1}$ is Lipschitz continuous with constant $k_{n-1}>0$. Consider two vectors $\mathbf{x}_1$ and $\mathbf{x}_2$ in $\mathbb{R}^{n}$, where the coordinates $x_{1,i}>1$ and $x_{2,i}>1, \forall i$. Without loss of generality, consider that $x_{1,n}\geq x_{2,n}$. Also, assume that the weights $w_{i} > 0$, $\forall{i}\in{1,2,\ldots,n}$. Using induction hypothesis, we have:
 \begin{eqnarray}
 |f_n(\mathbf{x}_1)-f_n(\mathbf{x}_2)| &=&  \left |\sum_{i=1}^n \log x_{1,i}^{w_i}-\sum_{i=1}^n \log x_{2,i}^{w_i}\right | \notag \\
 &\leq & \left|\sum_{i=1}^{n-1} \log x_{1,i}^{w_i}-\sum_{i=1}^{n-1} \log x_{2,i}^{w_i}\right | \notag \\
 && + \left |w_n\log \frac{x_{1,n}}{x_{2,n}} \right | \notag \\ 
 &\leq & k_{n-1}\sum_{i=1}^{n-1}|x_{1,i}-x_{2,i}| \notag \\
 && + w_n|x_{1,n}-x_{2,n}| \notag \\
 &\leq & k_n ||\mathbf{x}_1-\mathbf{x}_2|| \notag
 \end{eqnarray}
 where $k_n=\max(k_{n-1},w_n)$ and 
 consider the $L_1$ norm  in $\mathbb{R}^{n}$. This completes the proof.  
\end{proof}

\section{Proof of Theorem~\ref{th:2}} \label{app:poly_reduced}

\begin{proof} 
Let us assume one sub-carrier allocation is as such, says the first sub-carrier of the first BS. 
As aforementioned, there are at most $M$ users per sub-carrier due to the limitation of SIC. 
Without loss of generality, let us 
assume that the first user has the 
highest channel gain and the second user has the second highest channel gain. 
We use $R_{1}^{1}$ to denote the sum rate in the first sub-carrier of the above first BS, 
which is expressible as: 
\begin{eqnarray*} 
    R_{1}^{1} & = & \log{\Big(1 + \frac{g_{1,1}^{1} p_{1,1}^{1}}{\sum_{i=2}^{K}g_{i,1}^{1} p_{i}^{1} + N}}\Big) + ... + \\ & & \log\Big({1 + \frac{g_{1,M}^{1} p_{1,M}^{1}}{\sum_{j=1}^{M-1}g_{1,M}^{1} p_{1,j}^{1}+ \sum_{i=2}^{K}g_{i,2}^{1} p_{i}^{1} + N}}\Big)\\
    & = & \log\Big({\frac{g_{1,1}^{1} p_{1,1}^{1} + \sum_{i=2}^{K} g_{i,1}^{1} p_{i}^{1} + N}{g_{1,2}^{1} p_{1,1}^{1}+ \sum_{i=2}^{K}g_{i,2}^{1} p_{i}^{1} + N}}\Big)    + ... +\\ & & \log{\Big(\sum_{j=1}^{M}g_{1,M}^{1} p_{1,j}^{1}+ \sum_{i=2}^{K}g_{i,M}^{1} p_{i}^{1} + N  }\Big)- \\ & &  \log{\Big(\sum_{i=2}^{K} g_{i,1}^{1} p_{i}^{1} + N \Big) }\\ \\
    & = & \log\Big({\frac{g_{1,1}^{1} p_{1,1}^{1} + \sum_{i=2}^{K} g_{i,1}^{1} p_{i}^{1} + N}{g_{1,2}^{1} p_{1,1}^{1}+ \sum_{i=2}^{K}g_{i,2}^{1} p_{i}^{1} + N}}\Big)  + ... + \\ & & \log{\Big(g_{1,M}^{1} p_{1}^{1}+ \sum_{i=2}^{K}g_{i,2}^{1} p_{i}^{1} + N }\Big)- \\ & & \log{\Big(\sum_{i=2}^{K} g_{i,1}^{1} p_{i}^{1} + N\Big)}. 
\end{eqnarray*}
We now show that the first term of the above expression
is an increasing function of $p_{1,1}^{1}$. Let us define:  
\begin{equation*}
    f(p_{1,1}^{1}) \triangleq \log\Big({\frac{g_{1,1}^{1} p_{1,1}^{1} + \sum_{i=2}^{K} g_{i,1}^{1} p_{i}^{1} + N}{g_{1,2}^{1} p_{1,1}^{1}+ \sum_{i=2}^{K}g_{i,2}^{1} p_{i}^{1} + N}}\Big). 
\end{equation*} 
Taking the derivative of the above function with respect to  $p_{1,1}^{1}$, we have: 
\begin{eqnarray*}
     \frac{d(f(p_{1,1}^{1}))}{d p_{1,1}^{1}}  = &  \\
    & \hspace{-1.5cm}\frac{g_{1,1}^{1} (g_{1,2}^{1} p_{1,1}^{1}+ \sum_{i=2}^{K} g_{i,2}^{1} p_{i}^{1} + N) -  g_{1,2}^{1}(g_{1,1}^{1} p_{1,1}^{1} + \sum_{i=2}^{K} g_{i,1}^{1} p_{i}^{1} + N)}{(g_{1,1}^{1} p_{1,1}^{1} + \sum_{i=2}^{K} g_{i,1}^{1} p_{i}^{1} + N) \times (g_{1,2}^{1} p_{1,1}^{1}+ \sum_{i=2}^{K}g_{i,2}^{1} p_{i}^{1} + N)} 
\end{eqnarray*} 
which is always greater than zero because of the 
SIC 
condition \eqref{eq:sicconst1}. We can do the same analysis and reach the same conclusion  for all the other terms except the last term and the second last term. For a given $\mathbf{p}$, the second last term and the last term are constant. Thus, at 
each particular iteration of the  polyblock algorithm, if the power distribution over all the BS for a particular sub-carrier $l$ is $(p_{1}^{l}, \ldots, p_{K}^{l})$, the optimal choice is to allocate all the power to the best user in that sub-carrier. This completes the proof. 
\end{proof}

\end{appendices}


\end{document}